\documentclass[sn-mathphys-num]{sn-jnl}
\usepackage{graphicx}%
\usepackage{multirow}%
\usepackage{amsmath,amssymb,amsfonts}%
\usepackage{amsthm}%
\usepackage{mathrsfs}%
\usepackage[title]{appendix}%
\usepackage[dvipsnames]{xcolor}%
\usepackage{textcomp}%
\usepackage{manyfoot}%
\usepackage{booktabs}%
\usepackage{algorithm}%
\usepackage{algorithmicx}%
\usepackage{algpseudocode}%
\usepackage{listings}%
\usepackage[normalem]{ulem}

\usepackage{tikz}
\usetikzlibrary{automata, arrows.meta, positioning, shapes}
\usetikzlibrary{arrows.meta, 
                automata,
                intersections,
                positioning,
                quotes,
                matrix,
                bbox 
                }
\usetikzlibrary{shapes,arrows,fit}
\newtheorem{theorem}{Theorem}
\newtheorem{observation}[theorem]{Observation}
\newtheorem{definition}[theorem]{Definition}

\newtheorem{lemma}[theorem]{Lemma}
\newtheorem{corollary}[theorem]{Corollary}
\newtheorem{example}[theorem]{Example}

\providecommand{\N}{\mathbb{N}}
\providecommand{\mg}{\operatorname{mg}}
\providecommand{\lcp}{\operatorname{lcp}}
\providecommand{\lcs}{\operatorname{lcs}}

\providecommand{\cont}{\operatorname{content}}
\providecommand{\Fill}{\mathcal{F}}
\providecommand{\eLlen}{\varepsilon\mathcal{L}_{\mathrm{len}}}

\title{Learning Erasing Equal-Length Relational Pattern Languages in the Limit from Positive Data over Binary Alphabets}

\author*[1]{\fnm{S.\ Mahmoud} \sur{Mousawi}}\email{mousawi.s.m@gmail.com, seyyed.mousavi@uregina.ca}

\author*[1,2]{\fnm{Sandra} \sur{Zilles}}\email{sandra.zilles@uregina.ca}

\affil[1]{\orgdiv{Department of Computer Science}, \orgname{University of Regina}, \orgaddress{\country{Canada}}}

\affil[2]{\orgdiv{Alberta Machine Intelligence Institute}}

\abstract{Relational pattern languages generalize Angluin’s pattern languages by requiring the strings substituted for repeated occurrences of a variable to stand in a given relation, rather than to be identical. We study erasing pattern languages under the equal-length relation, where all occurrences of a variable must be replaced by strings of the same length. Over alphabets with at least three letters, this class is known to be learnable in the limit from positive data. We show that over a binary alphabet it is not, which resolves an open question from earlier work. Our proof reveals that non-learnability already holds for the restricted subclass of languages generated by patterns that have exactly three maximal substrings of variables. We then show that patterns with at most two maximal substrings of variables are learnable in the limit from positive data for every alphabet of size at least two. If the target pattern is known to have at most one maximal substring of variables, this can even be achieved by an iterative learner, which updates its hypothesis using only the current example and its previous hypothesis.}

\begin{document}

\maketitle

\section{Introduction}

Pattern languages, introduced by Angluin~\cite{Angluin80patterns}, have been an object of study in algorithmic learning theory for decades. A pattern is a finite string of terminal symbols and variables. Its language consists of all words obtained by replacing each variable with a string over the terminal alphabet, where all occurrences of the same variable are replaced by the same string. The literature distinguishes non-erasing pattern languages~\cite{Angluin80patterns}, in which variables cannot be replaced with the empty string, from erasing pattern languages, where empty substitutions are allowed~\cite{Shinohara82}. Both variants have been studied extensively in Gold’s model of learning in the limit from positive data~\cite{Gold67}, where a learner is gradually presented with a sequence of all strings in an unknown target language $L$ and must eventually settle on a correct description for $L$. The class of non-erasing pattern languages is learnable in this model~\cite{Angluin80patterns}, even by an iterative learner, which updates its hypothesis based only on the previous hypothesis and the current example~\cite{LangeW91}. By contrast, erasing pattern languages are not learnable in the limit from positive data, when the underlying alphabet has size 2, 3, or 4~\cite{Reidenbach06,Reidenbach08}.

Relational pattern languages~\cite{GeilkeZ11} are a generalization of the notion of pattern language, in which repeated variables need not be replaced by the same string, but instead by strings that are in a specified relation. Natural choices for this relation are, for example, equality, reversal, equal length, or abelian equivalence~\cite{HolteMZ22,MousawiZ25,JansenNPWW25}; the equality relation here recovers the original notion of pattern language.

The equal-length relation is arguably the simplest of the relations considered in the literature~\cite{HolteMZ22,MousawiZ25}. All occurrences of a variable must be replaced by strings of the same length, but the strings themselves can be chosen independently. As a result, the language generated by a pattern under this relation depends only on its terminal symbols and on the admissible lengths of the strings between them. 
Non-erasing equal-length relational pattern languages are trivially learnable in the limit from positive data; moreover it is known that they can be distinguished by small sets of short strings when the alphabet size is at least three~\cite{HolteMZ22}. In the erasing case, it was observed that the positive learnability result for alphabets with at least three letters carries over from the non-erasing case~\cite{MousawiZ25}. However, binary alphabets posed a challenge. It was shown that many specific forms of patterns generate erasing equal-length relational languages that have so-called \emph{telltales}\/ even for binary alphabets~\cite{MousawiZ25}. A telltale of a language $L$ with respect to a family $\mathcal{L}$ of languages is a finite subset of $L$ that is not contained in any language $L'\in\mathcal{L}$ with $L'\subset L$. The uniform enumerability of a system of telltales is a necessary and sufficient condition for the learnability of a family $\mathcal{L}$ with uniformly decidable membership~\cite{Angluin80telltales}. Despite numerous results on the existence of telltales for subclasses of erasing equal-length relational pattern languages, it was left open whether the full class $\varepsilon\mathcal{L}_{len}$ of erasing equal-length relational pattern languages has telltales and thus whether it is learnable in the limit from positive data, given a binary alphabet.

We resolve this open question by presenting a pattern $p^*$ whose erasing equal-length relational language $\varepsilon L_{len}(p^*)$ over a binary alphabet does not have a telltale with respect to $\varepsilon\mathcal{L}_{len}$. Consequently, for alphabets of only two letters, the class $\varepsilon\mathcal{L}_{len}$ is not learnable in the limit from positive data. We complement this result with positive learnability results for languages generated by patterns with a small number of maximal substrings of variables. 



The first technical ingredient is to represent every language $\varepsilon L_{len}(p)\in\eLlen$  by the sequence of terminals in the underlying pattern $p$ and by a finitely generated monoid of \emph{gap vectors}.\footnote{The notion of gap vectors is inspired by the notion of length tuples for equal-length relational patterns \cite{HolteMZ22}; we here extend this notion to study the corresponding monoid structure.} Roughly speaking, a word belongs to $\varepsilon L_{len}(p)$ if it can be cut around the terminals of $p$ so that the lengths of the pieces form a vector of the monoid. Suppose a learning algorithm is presented with positive examples only, i.e., words in the language $\varepsilon L_{len}(p)$, and has to identify $p$. Any single positive example shows a word, but not where it has to be cut. The question we study is how much of the monoid can nevertheless be recovered from the examples. Over an alphabet of three or more letters, every gap vector has a word that can be cut in only one way, so that positive examples reveal the monoid and we obtain learnability in the limit from positive examples. 
Our focus is the binary alphabet case. The main results of our paper are:
\begin{itemize}
 \item Over a binary alphabet, a word can sometimes be cut differently, when shifting two equal adjacent terminals by one position. This allows one variable of a pattern to be ``imitated'' by other variables, and learning from positive data becomes impossible (Section~\ref{sec:len-binary}). We prove this by providing a pattern $p$ such that $\varepsilon L_{len}(p)$ does not possess a telltale with respect to the class $\eLlen$. 
 \item The pattern $p$, as well as all the patterns that witness the fact that $p$ does not have a telltale, have exactly three \emph{variable blocks}, i.e., they can be factored into substrings of only terminal symbols and substrings of only variable symbols, such that exactly three substrings of variable symbols are used. We therefore obtain the corollary that the patterns with exactly three variable blocks generate a class of erasing equal-length relational pattern languages that is not learnable in the limit from positive data, when the alphabet size is 2.
 \item We contrast this by showing that the corresponding class is learnable in the limit from positive data, when the underlying patterns contain at most two variable blocks, regardless of the alphabet size (Section~\ref{sec:twoblocks}).
  \item If it is known that the underlying pattern $p$ has at most one variable block, we obtain an even stronger positive learnability result, again regardless of the alphabet size. Learning in the limit from positive data can then even be achieved by an \emph{iterative}\/ learner, which updates its hypothesis based only on the previous hypothesis and the most recently seen word, ignoring all previously collected data (Section~\ref{sec:len-oneblock}).
\end{itemize}

\section{Preliminaries}
\label{sec:len-prelim}

\subsection{Equal-Length Patterns and Their Languages}

Given a set $W$, we denote the set of finite strings over $W$ by $W^*$ and its subset of non-empty finite strings by $W^+$. If $w\in W^*$ is any string over $W$, then we denote by $|w|$ the length of $w$. We can then write $w=w[1]w[2]\cdots w[|w|]$ for $w[1],w[2],\ldots,w[|w|]\in W$, where $w[i]$ denotes the symbol at the $i$th position in $w$. Moreover, for any $i,j\in\{1,\ldots,|w|\}$, the term $w[i:j]$ refers to the substring $w[i]w[i+1]\ldots w[j]$ of $w$ if $i\le j$, and to the empty string $\varepsilon$ if $i>j$. Likewise, $w[i:j]:=\varepsilon$ whenever $j=0$.

Throughout this document, we assume that $\Sigma$ is a finite alphabet of cardinality at least 2, and $X=\{x_1,x_2,x_3,\ldots\}$ is a countably infinite set of variables, disjoint from $\Sigma$. A pattern is a finite non-empty string over $\Sigma\cup X$, with the property that variable $x_i$ may occur in $p$ only if all variables $x_j$, for $1\le j<i$ also occur in $p$. If $p$ is a pattern containing all and only the variables $x_1,\ldots,x_i$ for some $i\ge 1$, and variable $x_j$ occurs exactly $n_j$ ($\ge 1$) times in $p$, then an \emph{equal-length substitution}\/ for $p$ is a mapping defined by $$x_j\mapsto (w_{j,1},w_{j,2},\ldots,w_{j,n_j})\mbox{ for }1\le j\le i\,,$$
where $w_{j,1},w_{j,2},\ldots,w_{j,n_j}\in\Sigma^*$ and $$|w_{j,1}|=|w_{j,2}|=\cdots =|w_{j,n_j}|\mbox{ for }1\le j\le i\,.$$
This substitution generates from $p$ the string that is obtained if the $k$th occurrence of the variable $x_j$ in $p$ (counted from the left) is replaced by $w_{j,k}$, for each $j\in\{1,\ldots,i\}$ and each $k\in\{1,\ldots,n_j\}$. The term \emph{equal-length}\/ refers to the fact that such a substitution replaces all occurrences of the same variable with strings of the same length. This stands in contrast with the original definition of pattern languages, where multiple occurrences of the same variable have to be substituted by the same string~\cite{Angluin80patterns,Shinohara82}. If $p$ has no variables, i.e., $p\in\Sigma^+$, then the only word that can be generated from $p$ is the word $p$ itself. 

For example, consider the pattern $p=x_1\,a\,a\,x_1\,x_2\,b$ and the equal-length substitution defined by $x_1\mapsto(ab,aa)$ and $x_2\mapsto(b)$. This substitution generates the string $ab\,a\,a\,aa\,b\,b$ from $p$. The (erasing) equal-length relational language of the pattern $p$ (over $\Sigma$) is the set of all strings that are generated from $p$ by substitutions. 

Earlier work also considered \emph{non-erasing}\/ equal-length relational pattern languages, which are defined in the same way, but using only substitutions for which the words $w_{j,1},w_{j,2},\ldots,w_{j,n_j}$ belong to $\Sigma^+$, i.e., are non-empty~\cite{HolteMZ22}. Aside from equal-length relational pattern languages, other relations between the strings replacing multiple occurrences of the same variables have been considered in the literature, including reversal, abelian equivalence, et cetera~\cite{GeilkeZ11,HolteMZ22,JansenNPWW25}. From now on, we consider exclusively the equal-length relation and focus on (erasing) equal-length relational pattern languages.

Any pattern $p$ has a unique decomposition of the form
\[
  p=\pi_0\,c_1\,\pi_1\,c_2\cdots c_m\,\pi_m,\qquad c_k\in\Sigma,\ \pi_j\in X^\ast .
\]
The non-empty $\pi_j$ are called the \emph{variable blocks}\/ of $p$. We call $s(p)=c_1\cdots c_m$ the \emph{terminal word} of $p$ and $\pi_0,\dots,\pi_m$ its \emph{gaps}. For instance, $p=x_1\,a\,a\,x_1\,x_2\,b$ has the gaps $\pi_0=x_1$, $\pi_1=\varepsilon$, $\pi_2=x_1\,x_2$ and $\pi_3=\varepsilon$. Its variable blocks are $x_1$ and $x_1\, x_2$.

For a variable $x$ in $p$ we define its \emph{gap vector} $\gamma(x)\in\N^{m+1}$ by letting $\gamma(x)_j$ be the number of occurrences of $x$ in the gap $\pi_j$. For a vector $v\in\N^{m+1}$ we write $|v|=v_0+\dots+v_m$ for its \emph{size}, and for a set $P\subseteq\N^{m+1}$ we write $\langle P\rangle$ for the submonoid of $\N^{m+1}$ generated by $P$.
For a submonoid $\langle P\rangle\subseteq\N^{m+1}$ let $\mg(\langle P\rangle)$ be the set of minimal generators of $\langle P\rangle$, i.e., the set of non-zero elements of $\langle P\rangle$ that are not the sum of two non-zero elements of $\langle P\rangle$; here $\mg(\{0\}):=\emptyset$. The set $\mg(\langle P\rangle)$ generates $\langle P\rangle$ and is always contained in $\langle P\rangle$. For a finite set $A$ of integers, we use $\mg(A) := mg(\langle A\rangle)$. 

\medskip

\begin{definition}
Let $s=c_1\cdots c_m$. A \emph{split} of a word $w$ (with respect to $s$) is a factorization $w=w_0c_1w_1\cdots c_mw_m$ with $w_j\in\Sigma^\ast$. 
Its \emph{vector} is $(|w_0|,\dots,|w_m|)$, and $D_s(w)$ denotes the set of vectors of all splits of $w$. 
The \emph{fillings} of a vector $v\in\N^{m+1}$ with respect to $s$ are the words in the set
\[
  \Fill_s(v)=\{w_0c_1w_1\cdots c_mw_m \mid w_j\in\Sigma^\ast,\ |w_j|=v_j\}\,.
\]
For a set $P\subseteq\N^{m+1}$ we let
\[
  L_s(P)=\{w\in\Sigma^\ast \mid D_s(w)\cap\langle P\rangle\neq\emptyset\}=\bigcup\nolimits_{v\in\langle P\rangle}\Fill_s(v).
\]
\end{definition}

The following observation, expressing the equal-length relational language of a pattern $p$ in terms of the fillings of its gap vectors, is immediate.

\medskip

\begin{observation}[Normal form]\label{lem:len-normal}
Let $p$ be a pattern with terminal word $s$, and let $P$ be the set of gap vectors of its variables. Then
\[
  \varepsilon L_{\mathrm{len}}(p)=L_s(P).
\]
Conversely, for every word $s$ and every finite set $P\subseteq\N^{|s|+1}\setminus\{0^{|s|+1}\}$ there is a pattern $p$ with $\varepsilon L_{\mathrm{len}}(p)=L_s(P)$. 
\end{observation}

\medskip

\begin{example}\label{ex:len-splits}
Let $p=x_1\,a\,x_2x_2\,x_1\,b\,x_2\,x_3$. Its terminal word is $ab$ and the gap vectors are $(1,1,0)$, $(0,2,1)$ and $(0,0,1)$. The substitution $x_1\mapsto (a,a)$, $x_2\mapsto (b,a,a)$, $x_3\mapsto(\varepsilon)$ generates $aabaaba$; the terminals of $p$ give the split $a\cdot a\cdot baa\cdot b\cdot a$ with vector $(1,3,1)=(1,1,0)+(0,2,1)$. The same word has five further splits, with vectors $(0,1,4)$, $(1,0,4)$, $(0,4,1)$, $(3,1,1)$ and $(4,0,1)$. A word does not reveal which split was used to generate it; this ambiguity turns out to be the source of difficulties for learning equal-length relational pattern languages from positive examples (for binary alphabets).
\end{example}

\subsection{Learning in the Limit and Telltales}

The focus of this paper is on learning equal-length relational pattern languages in Gold's model of learning in the limit from positive data \cite{Gold67}.

A \emph{text} for $L$ is an infinite sequence $t_0,t_1,\dots$ with $\{t_i\mid i\in\N\}=L$. 

A learner is a mapping that takes as input any initial segment $t_0,t_1,\ldots,t_n$ of any text of any target language. If defined, its output is a representation of some language $L'$, which is to be understood as the learner hypothesizing the target language to be $L'$. In general, the goal of a learner is to eventually identify an unknown language $L$ in a known class $\mathcal{L}$ of languages, when gradually processing any text for $L$. 

\medskip

\begin{definition}
    Let $\mathcal{L}$ be a class of languages and $L\in\mathcal{L}$. Fix a numbering $(L_i)_{i\in\mathbb{N}}$ such that $\mathcal{L}\subseteq\{L_i\mid i\in\mathbb{N}\}$. A learner $A$ learns $L$ in the limit from positive data (with respect to $(L_i)_{i\in\mathbb{N}}$) if, for any text $t_0,t_1,t_2,\ldots$ for $L$, the following two properties hold:
    \begin{itemize}
        \item $A(t_0,\ldots,t_n)$ is defined for every $n\in\mathbb{N}$.
        \item There exists some $n_0\in\mathbb{N}$ such that $L_{A(t_0,\ldots,t_{n_0})}=L$ and $A(t_0,\ldots,t_n)=A(t_0,\ldots,t_{n_0})$ for every $n\ge n_0$.
    \end{itemize}
    The class $\mathcal{L}$ is learnable in the limit from positive data if there is a numbering $(L_i)_{i\in\mathbb{N}}$ and a learner that learns every $L\in\mathcal{L}$ from positive data with respect to  $(L_i)_{i\in\mathbb{N}}$.
\end{definition}

\medskip

Angluin's characterization of learnability in the limit from positive data \cite{Angluin80telltales} concerns computable learners for indexed families $(L_i)_{i\in\mathbb{N}}$; these are families for which a recursive procedure exists that, given $(i,w)\in\mathbb{N}\times\Sigma^*$, decides whether $w\in L_i$. Clearly, the class of all equal-length relational pattern languages is an indexed family, since its membership problem is uniformly decidable. The central notion in Angluin's characterization is that of a telltale: A \emph{telltale} of a language $L$ in a class $\mathcal{L}$ is a finite set $T\subseteq L$ such that no $L'\in\mathcal{L}$ satisfies $T\subseteq L'\subsetneq L$.

Angluin's result states that an indexed family $(L_i)_{i\in\mathbb{N}}$ of languages is learnable in the limit from positive data if and only if there is an effective procedure that, given an index $i$, enumerates a telltale for $L_i$ with respect to $\{L_i\mid i\in\mathbb{N}\}$. Note that Angluin's characterization builds on the same idea as the \emph{locking sequences} of Blum and Blum~\cite{BlumBlum1975}. For our negative result we need only the necessity part of the characterization, and we do not require computable learners or the recursive enumerability of telltales, i.e., the following  statement suffices for our purposes: 

\medskip

\begin{theorem}[\cite{Angluin80telltales,BlumBlum1975}]\label{thm:len-angluin}
If some (computable or non-computable) learner learns an indexed family $\mathcal{L}$ of languages in the limit from positive data, then every $L\in\mathcal{L}$ has a telltale in~$\mathcal{L}$.
\end{theorem}

\medskip

In the context of this paper, we focus on learning equal-length relational pattern languages in the limit from positive data. Hypotheses output by the learner are (indices of) relational patterns in a fixed enumeration.


\section{Non-Learnability over Binary Alphabets}
\label{sec:len-binary}

In this section, let $\Sigma=\{a,b\}$. Reidenbach showed that the erasing pattern languages under equality relation are not learnable in the limit from positive data over small alphabets~\cite{Reidenbach06,Reidenbach08}, and Mousawi and Zilles~\cite{MousawiZ25} obtained the same result for the reversal relation over a binary alphabet. We now show that the equal-length relation behaves in the same way, even though its languages are determined by the lengths of the gaps alone.

We exhibit a language $L\in\eLlen$ and a chain of languages $L_0\subsetneq L_1\subsetneq L_2\subsetneq\cdots$ in $\eLlen$ whose union is $L$. Every finite subset of $L$ is then contained in some $L_t\subsetneq L$, so $L$ has no telltale. Therefore, by Theorem~\ref{thm:len-angluin}, no learner can identify all of these languages. 

The terminal word for our target language $L$ is $s=aab$. Gap vectors therefore have four coordinates, counting the variables before the first $a$, between the two $a$'s, between the second $a$ and $b$, and after $b$. Let
\[
  X_1=(1,0,1,0),\qquad X_2=(0,0,2,0),\qquad X_3=(0,0,0,1).
\]
The target is
\[
  L=\varepsilon L_{\mathrm{len}}(x_1\,a\,a\,x_1\,x_2x_2\,b\,x_3)=L_{aab}(\{X_1,X_2,X_3\})\,.
\]
For $t\ge0$ we let
\[
  L_t=L_{aab}\bigl(\{X_3\}\cup\{X_1+jX_2\mid 0\le j\le t\}\bigr).
\]
Since $X_1+jX_2=(1,0,2j+1,0)$, a pattern for $L_t$ has, for each $j\le t$, a  variable occurring once before $aa$ and $2j+1$ times between the second $a$ and $b$, in addition to the variable $y_3$ at the end of the pattern. For example,
\[
  L_0=\varepsilon L_{\mathrm{len}}(y_1\,a\,a\,y_1\,b\,y_3),\qquad
  L_1=\varepsilon L_{\mathrm{len}}(y_1\,y_2\,a\,a\,y_1\,y_2y_2y_2\,b\,y_3)\,.
\]
 The target pattern, as well as all patterns generating the languages $L_t$ for $t\in\mathbb{N}$, have three variable blocks each. The target has the variable $x_2\ne x_1$ occurring only in the middle variable block; the languages $L_t$ have every variable in the middle block also appear in the first variable block.

 For the remainder of this section we use the following notation:
 $$M=\langle X_1,X_2,X_3\rangle$$ and $$M_t=\langle X_3,\ X_1+jX_2\mid 0\le j\le t\rangle\,.$$
In particular, $M$ and $M_t$ are submonoids of $\mathbb{N}^4$.

We first make a simple observation on the structure of the monoids representing our target language $L$ and the languages $L_t$ in the chain.

\medskip

\begin{lemma}\label{lem:len-binary-monoids}
$M$ and $M_t$ satisfy the following:
\begin{align*}
  M   &=\{(\alpha,0,y,\delta)\mid \alpha,y,\delta\in\N,\ y\ge\alpha,\ y\equiv\alpha \pmod2\},\\
  M_t &=\{(\alpha,0,y,\delta)\mid \alpha,y,\delta\in\N,\ \alpha\le y\le(2t+1)\alpha,\ y\equiv\alpha \pmod2\}.
\end{align*}
In particular, every vector of $M$ and of $M_t$ has second coordinate $0$, and every vector of $M_t$ with first coordinate $0$ also has third coordinate $0$.
\end{lemma}

\begin{proof}
Any element of $M$ is of the form $\alpha X_1+\beta X_2+\delta X_3=(\alpha,0,\alpha+2\beta,\delta)$ with $\alpha,\beta,\delta\in\N$. Its third coordinate $y=\alpha+2\beta$ satisfies $y\ge\alpha$ and $y\equiv\alpha\pmod2$. Conversely, a vector $(\alpha,0,y,\delta)$ with $y\ge\alpha$ and $y\equiv\alpha\pmod2$ belongs to $M$, as witnessed by $\beta=(y-\alpha)/2$.

Any element of $M_t$ is of the form $\delta X_3$ plus a sum of $\alpha$ vectors $X_1+j_iX_2=(1,0,2j_i+1,0)$ with $0\le j_i\le t$. Its first coordinate is $\alpha$, and its third coordinate $y=\sum_{i=1}^{\alpha}(2j_i+1)$ is a sum of $\alpha$ odd numbers between $1$ and $2t+1$. Hence $\alpha\le y\le(2t+1)\alpha$ and $y\equiv\alpha\pmod 2$. Conversely, let $y$ satisfy $\alpha\le y\le(2t+1)\alpha$ and $y\equiv\alpha\pmod 2$ for $\alpha\in\mathbb{N}$. Clearly, then $(\alpha,0,y,\delta)\in M_t$ for all $\delta\in\mathbb{N}$.
\end{proof}


\medskip

It is straightforward to obtain the following chain.

\medskip

\begin{lemma}\label{lem:len-binary-incl}
$L_0\subseteq L_1\subseteq L_2\subseteq\cdots\subseteq L$.
\end{lemma}

\begin{proof}
Every generator of $M_t$ is a generator of $M_{t+1}$, and $X_1+jX_2\in M$. Hence $M_0\subseteq M_1\subseteq\cdots\subseteq M$, from which the claim follows immediately.
\end{proof}

Next, we show that every word in $L$ is contained in some language $L_t$, $t\in\mathbb{N}$, so that, together with Lemma~\ref{lem:len-binary-incl}, we obtain that the union of all languages $L_t$ in the chain equals $L$.

\medskip

\begin{lemma}\label{lem:len-binary-union}
Let $w\in L$ have a split with vector $(\alpha,0,\alpha+2\beta,\delta)\in M$. Then $w\in L_\beta$. Consequently, $\bigcup_{t\in\N}L_t=L$.
\end{lemma}

\begin{proof}
Let $e=(\alpha,0,\alpha+2\beta,\delta)$. We find a split of $w$ whose vector lies in $M_\beta$, distinguishing three cases.

\emph{Case $\alpha\ge1$.} Then
\[
  e=(\alpha-1)\,X_1+(X_1+\beta X_2)+\delta X_3,
\]
since $(\alpha-1)(1,0,1,0)+(1,0,1+2\beta,0)+(0,0,0,\delta)=(\alpha,0,\alpha+2\beta,\delta)$. All three summands belong to $M_\beta$, so the given split already works.

\emph{Case $\alpha=0$ and $\beta=0$.} Then $e=\delta X_3\in M_0\subseteq M_\beta$. 

\emph{Case $\alpha=0$ and $\beta\ge1$.} The given split has the form $w=aa\,u\,b\,v$ with $|u|=2\beta\ge2$ and $|v|=\delta$. It uses the variable $x_2$ only in the middle variable block, corresponding to a vector $(0,0,2\beta,\delta)$ which is not in $M_\beta$, so we split $w$ differently, depending on the first letter of $u$.
\begin{itemize}
  \item If $u$ begins with $b$, then $w$ begins with $aab$. We use these three letters as the terminals. The resulting split has vector $(0,0,0,|w|-3)=(0,0,0,2\beta+\delta)=(2\beta+\delta)X_3\in M_0\subseteq M_\beta$. 
  \item If $u$ begins with $a$, say $u=a\,u'$, then $w=a\,aa\,u'\,b\,v$. We use the first letter of $w$ as the content of the first gap, its second and third letters as the two terminals $a$, and the same $b$ as before. The second gap is empty, the third gap is $u'$ of length $2\beta-1$, and the fourth gap is $v$. The resulting split has vector
  \[
    (1,0,2\beta-1,\delta)=\bigl(X_1+(\beta-1)X_2\bigr)+\delta X_3\in M_{\beta-1}.
  \]
\end{itemize}
In every case, $w\in L_\beta$, using $L_{\beta-1}\subseteq L_\beta$ in the last case. Since every word of $L$ has a split with vector in $M$, every word of $L$ belongs to some $L_t$, and together with Lemma~\ref{lem:len-binary-incl} we obtain $\bigcup_tL_t=L$.
\end{proof}

\begin{example}
The word $aabaabb=aa\cdot baab\cdot b$ has the split vector $(0,0,4,0)=2X_2$ in $M$. Since $baab$ begins with $b$, we may use its first occurrence of $aab$ as the terminals, which gives the vector $(0,0,0,4)=4X_3\in M_0$; so this word belongs to $L_0$. The word 
$aaababb=aa\cdot abab\cdot b$ also has the split vector $2X_2$. Since $abab$ begins with $a$, we shift the terminals $aa$ one position to the right, obtaining $a\cdot aa\cdot bab\cdot b$
with vector $(1,0,3,0)=X_1+X_2\in M_1$; so this word belongs to $L_1$. (Lemma~\ref{lem:len-binary-strict} below will show that this word does not belong to $L_0$.)
\end{example}

\medskip

Finally, we demonstrate that the chain is strict.

\medskip

\begin{lemma}\label{lem:len-binary-strict}
For every $t\ge0$, the word
\[
  w_t=aaab\,(ab)^{t+1}\,b\qquad\text{of length }2t+7
\]
belongs to $L_{t+1}\setminus L_t$. Consequently, $L_t\subsetneq L_{t+1}$ and $L_t\subsetneq L$.
\end{lemma}

\begin{proof}
We number the positions of $w_t$ from $0$. The letter $a$ occurs at the positions $0,1,2$ and at the even positions $4,6,\dots,2t+4$; the letter $b$ occurs at position $3$, at the odd positions $5,7,\dots,2t+5$, and at position $2t+6$. For instance, $w_0=aaababb$.

By Lemma~\ref{lem:len-binary-monoids}, a split whose vector lies in $M_{t+1}$ or in $M_t$ uses two adjacent occurrences of $a$ as the terminals $a,a$. After position $3$ in $w_t$, the letters $a$ and $b$ alternate, so the only adjacent pairs of $a$'s are at the position pairs $(0,1)$ and $(1,2)$. Let $r$ be the position used for the terminal $b$. A split using the pair $(q,q+1)$ and $r$ has the vector $(q,\,0,\,r-q-2,\,2t+6-r)$. Table~\ref{tab:splits} lists all possibilities.

\begin{table}[t]
\begin{tabular}{|l|l|l|p{0.36\textwidth}|}
\hline
pair of $a$'s & position $r$ of $b$ & vector & membership\\
\hline
$(0,1)$ & any $r\ge3$ & $(0,0,r-2,\cdot)$ & not in $M_t$: first coordinate $0$, third coordinate at least $1$\\
$(1,2)$ & $r=3$ & $(1,0,0,\cdot)$ & neither in $M$ nor in $M_t$: third coordinate smaller than the first\\
$(1,2)$ & $r$ odd, $5\le r\le 2t+5$ & $(1,0,r-3,\cdot)$ & neither in $M$ nor in $M_t$: wrong parity\\
$(1,2)$ & $r=2t+6$ & $(1,0,2t+3,0)$ & in $M_{t+1}$, as $X_1+(t+1)X_2$; not in $M_t$, as $2t+3>2t+1$\\
\hline
\end{tabular}
\caption{Possible combinations for positioning terminal symbols in $w_t$, and resulting membership of the split vector in $M_{t+1}$ and $M_t$.}\label{tab:splits}
\end{table}

The last row shows $w_t\in L$, and no row yields a vector of $M_t$, so $w_t\notin L_t$. By Lemma~\ref{lem:len-binary-union}, applied to the split of the last row with $\beta=t+1$, we have $w_t\in L_{t+1}$. Hence $L_t\neq L_{t+1}$, and with Lemma~\ref{lem:len-binary-incl} the chain is strict.
\end{proof}

For example, for $t=0$ we obtain the following: $w_0=aaababb$ has its adjacent $a$'s at position pairs $(0,1)$ and $(1,2)$ and its $b$'s at positions $3,5,6$. The pair $(0,1)$ gives the vectors $(0,0,1,3)$, $(0,0,3,1)$ and $(0,0,4,0)$, none of which lies in $M_0$. The pair $(1,2)$ gives $(1,0,0,3)$, $(1,0,2,1)$ and $(1,0,3,0)$; the first two are not even in $M$, and the third is $X_1+X_2\in M_1\setminus M_0$.

\medskip

We finally put everything together to obtain our main result.

\medskip

\begin{theorem}\label{thm:len-no-telltale}
Let $\Sigma=\{a,b\}$. Then the language $L=\varepsilon L_{\mathrm{len}}(x_1\,a\,a\,x_1\,x_2x_2\,b\,x_3)$ has no telltale in $\eLlen$. In particular, the class $\eLlen$ is not learnable in the limit from positive data, not even by a non-computable learner.
\end{theorem}

\begin{proof}
Let $T\subseteq L$ be finite. By Lemma~\ref{lem:len-binary-union}, every word of $T$ belongs to some $L_t$; let $t$ be the largest index needed. By Lemma~\ref{lem:len-binary-incl}, $T\subseteq L_t$, and by Lemma~\ref{lem:len-binary-strict}, $L_t\subsetneq L$. Since $L_t\in\eLlen$, the set $T$ is not a telltale of $L$.
\end{proof}

The binary alphabet is used at exactly one place, in the last case of the proof of Lemma~\ref{lem:len-binary-union}: the first letter of $u$ is either $b$, in which case it can serve as the terminal $b$, or $a$, in which case the pair of terminals $aa$ can be shifted one position to the right. Over a ternary alphabet $\{a,b,c\}$, a word $aa\,c\,u'\,b\,v$ admits neither alternative. 

The shift needs two equal adjacent terminals, forming an empty gap, followed by a gap between two different terminals. The large class of patterns studied in \cite{MousawiZ25}, for which telltales exist over binary alphabets, has no such configuration. This raises the following question.

\emph{Given $\Sigma=\{a,b\}$, is the subclass of $\eLlen$ consisting of the languages whose terminal word does not contain $aa$ or $bb$ as a factor learnable in the limit from positive data?
}

We leave this question open.

Observe that the language $L$ from Theorem~\ref{thm:len-no-telltale}, as well as all the languages $L_t$ in the chain in Lemma~\ref{lem:len-binary-incl}, are generated by patterns with exactly three variable blocks. We therefore obtain the following corollary.

\medskip

\begin{corollary}\label{cor:threeblocks}
    Let $|\Sigma|=2$. Then the subclass of $\eLlen$ consisting of languages generated by patterns with exactly three variable blocks is not learnable in the limit from positive data.
\end{corollary}

\medskip

This immediately raises the question whether patterns with up to two variable blocks are learnable, for binary alphabets. The next section answers this question affirmatively.

\section{Learning of Patterns with at Most Two Variable Blocks}\label{sec:twoblocks}

Let $\eLlen^{2}$ be the class of the languages $\varepsilon L_{\mathrm{len}}(p)$ of patterns of the form \begin{equation}\label{eqn:patternform}p=\omega_0\,\vec{x}_1\,\omega_1\,\vec{x}_2\,\omega_2\end{equation} with $\omega_0,\omega_1,\omega_2\in\Sigma^\ast$, $\vec{x}_1,\vec{x}_2\in X^\ast$, and $p\ne\varepsilon$, i.e., patterns with at most two variable blocks. We require $\omega_1\ne\varepsilon$ when both $\vec{x}_1\ne\varepsilon$ and $\vec{x}_2\ne\varepsilon$. Moreover, when $p$ has exactly one variable block, we take $\vec{x}_2=\omega_2=\varepsilon$. This section establishes the following result, which stands in contrast with Corollary~\ref{cor:threeblocks}.

\medskip

\begin{theorem}\label{thm:twoblocks}
Let $|\Sigma|\ge 2$. Then $\eLlen^{2}$ is learnable in the limit from positive data.
\end{theorem}

\medskip

The remainder of this section is dedicated to proving Theorem~\ref{thm:twoblocks}.

Following Observation~\ref{lem:len-normal}, we identify the erasing equal-length relational language of a pattern $p$ of the form~(\ref{eqn:patternform}) with the monoid $M=\langle P\rangle$ generated by the set $P$ of gap vectors of the variables in $p$. In particular, 
\begin{equation}\label{eqn:M}
  \varepsilon L_{\mathrm{len}}(p)=\{\omega_0\,u\,\omega_1\,v\,\omega_2\mid u,v\in\Sigma^\ast,\ (|u|,|v|)\in M\}.
\end{equation}
Note that, by contrast with the previous section, within this section, $M$ is a submonoid of $\mathbb{N}^2$.

We begin with two straightforward lemmas.

\medskip

\begin{lemma}\label{lem:len-unique-occ}
Let $\omega_1\in\Sigma^+$ and let $\sigma_1, \sigma_2 \in\Sigma$ such that $\sigma_1 \neq \omega_1 [1]$ and $\sigma_2 \neq \omega_1 [|\omega_1|]$. For all $k_1,k_2 \in\N$, the word $w = \sigma_1^{k_1}\,\omega_1\,\sigma_2^{k_2}$ contains exactly one occurrence of the substring $\omega_1$, i.e., $k_1+1$ is the only integer  such that \( w[k_1+1: k_1+|\omega_1|] = \omega_1 \).
\end{lemma}

\begin{proof}
    Let $d$ be a positive integer such that \( w[d: d+|\omega_1|-1] = \omega_1 \). Since $\sigma_1 \neq \omega_1[1]$, we have $d \geq k_1+1$. Since $\omega_1$ does not end with $\sigma_2$, we also have $d \leq k_1+1$. 
\end{proof}

\begin{lemma}\label{lem:len-twoblock-witness}
Let $p$ be as in~(\ref{eqn:patternform}), let $M$ be as in~(\ref{eqn:M}). 
If $\omega_1\ne\varepsilon$, choose letters $\sigma_1 \neq \omega_1[1]$ and $\sigma_2 \neq\omega_1[|\omega_1|]$; otherwise choose $\sigma_1,\sigma_2\in\Sigma$ arbitrarily.
Then $\omega_0\,\sigma_1^{k_1}\,\omega_1\, \sigma_2^{k_2}\,\omega_2\in \varepsilon L_{len}(p)$ if and only if $(k_1,k_2)\in M$.
\end{lemma}

\begin{proof}
The ``if'' direction is immediate from the definition of equal-length relational pattern languages. The ``only if'' direction follows likewise with the help of Lemma~\ref{lem:len-unique-occ}, since every word in $\varepsilon L_{len}(p)$ is of the form $\omega_0\,u\,\omega_1\, v\,\omega_2$ for $(|u|,|v|)\in M$. Lemma~\ref{lem:len-unique-occ} forces $|u|=k_1$, so that $|v|=k_2$, implying $(k_1,k_2)\in M$.
\end{proof}

The next lemma shows that there are constraints on the placement of variable blocks in any pattern $p'$ whose language is contained in that of a pattern $p$ of form~(\ref{eqn:patternform}) and which generates the terminal word $\omega_0\omega_1\omega_2$ of $p$. 

\medskip

\begin{lemma}\label{lem:len-twoblock-positions}
Let \( p = \omega_0 \vec{x}_1 \omega_1 \vec{x}_2 \omega_2 \) be as in~(\ref{eqn:patternform}), and let \( L'\in\eLlen^{2} \) satisfy \( \omega_0\omega_1\omega_2 \in L' \subseteq \varepsilon L_{len} (p) \). Let $p'$ be a pattern with at most two variable blocks that generates $L'$. Then  $p'=\omega'_0\vec{x}'_1\omega'_1\vec{x}'_2\omega'_2$ where $\omega'_0,\omega'_1,\omega'_2\in\Sigma^\ast$, $\vec{x}'_1,\vec{x}'_2\in X^\ast$, and $\omega'_0\omega'_1\omega'_2=\omega_0\omega_1\omega_2$. Moreover, the following two statements hold.
\begin{enumerate}
    \item If $\vec{x}'_1\ne\varepsilon$, then $|\omega_0|\le|\omega'_0|\le|\omega_0\omega_1|$.
    \item If $\vec{x}'_2\ne\varepsilon$, then $|\omega_2|\le|\omega'_2|\le|\omega_1\omega_2|$.
\end{enumerate}
In particular, $p'$ has the prefix $\omega_0$ and the suffix $\omega_2$.
\end{lemma}

\begin{proof}
Since $L'\subseteq \varepsilon L_{len}(p)$, every word of $L'$ has length at least $|\omega_0\omega_1\omega_2|$; as $\omega_0\omega_1\omega_2\in L'$, the word $\omega_0\omega_1\omega_2$ is the shortest word in $L'$, and thus the terminal word of $p'$, by Observation~\ref{lem:len-normal}. It remains to show statements 1 and 2. We prove only statement 1; statement 2 follows by a completely symmetric argument.

Suppose $\vec{x}'_1\ne\varepsilon$, but $|\omega_0|>|\omega'_0|$ or $|\omega'_0|>|\omega_0\omega_1|$. 

If $|\omega_0|>|\omega'_0|$ (so that, in particular, $\omega_0\ne\varepsilon$), then let $\sigma\in\Sigma$ be a letter distinct from the last letter in $\omega_0$. Now consider a substitution for $p'$ that replaces every occurrence of every variable in $\vec{x}'_1$ with $\sigma^{|\omega_0|}$ and erases all other variables. This substitution generates a word in $\varepsilon L_{len}(p')$ ($=L'$) that does not begin with $\omega_0$. Since every word in $\varepsilon L_{len}(p)$ begins with $\omega_0$, this yields a contradiction to $L'\subseteq\varepsilon L_{len}(p)$.

If $|\omega'_0|>|\omega_0\omega_1|$, then $|\omega'_1\omega'_2|<|\omega_2|$, and we obtain a contradiction by considering a substitution for $p'$ that replaces every variable in $p'$ with $\tau^{|\omega_2|}$, where $\tau$ is a letter distinct from the first letter in $\omega_2$. Such a substitution generates a word that does not end with $\omega_2$, unlike all words in $\varepsilon L_{len}(p)$.
\end{proof}

\begin{lemma}\label{lem:len-twoblock-telltale}
Let $|\Sigma| \geq 2$, let $p=\omega_0 \vec{x}_1 \omega_1 \vec{x}_2 \omega_2$ be as in~(\ref{eqn:patternform}), and let $M$ be as in~(\ref{eqn:M}). If $\omega_1\ne\varepsilon$, choose letters $\sigma_1 \neq \omega_1[1]$ and $\sigma_2 \neq\omega_1[|\omega_1|]$; otherwise choose $\sigma_1,\sigma_2\in\Sigma$ arbitrarily. Then
\[
  T:=\{\omega_0\omega_1\omega_2\}\cup\{\omega_0\,\sigma_1^{k_1}\,\omega_1\, \sigma_2^{k_2}\,\omega_2 \mid (k_1,k_2) \in \mg(M)\}
\]
is a telltale of $\varepsilon L_{len}(p)$ with respect to $\eLlen^{2}$.
\end{lemma}

\begin{proof}
Note that the letters $\sigma_1$ and $\sigma_2$ satisfy the assumption of Lemma~\ref{lem:len-unique-occ}. Thus, $T\subseteq \varepsilon L_{len}(p)$ follows from  Lemma~\ref{lem:len-twoblock-witness}.

Let $L'\in\eLlen^{2}$ with $T \subseteq L' \subseteq \varepsilon L_{len}(p)$. We will show $L' = \varepsilon L_{len}(p)$. The latter is immediate if $\vec{x}_1=\vec{x}_2=\varepsilon$, since then $T=L'=\varepsilon L_{len}(p)=\{\omega_0\omega_1\omega_2\}$. So suppose $p$ has at least one non-empty variable block, so that $\varepsilon L_{len}(p)$ is infinite and $|T|>1$.

Let $p'=\omega'_0\vec{x}'_1\omega'_1\vec{x}'_2\omega'_2$ be a pattern with at most two variable blocks such that $L'=\varepsilon L_{len}(p')$. By Lemma~\ref{lem:len-twoblock-positions}, we know $\omega'_0\omega'_1\omega'_2=\omega_0\omega_1\omega_2$, and $p'$ has the prefix $\omega_0$ and the suffix $\omega_2$. Moreover, since $T\subseteq \varepsilon L_{len}(p')$, the pattern $p'$ has at least one (non-empty) variable block.

We prove two crucial claims.

\smallskip

\emph{Claim 1: If $\vec{x}_1\ne\varepsilon$, then $p'$ starts with $\omega_0x$ for some $x\in X$.}  Since $\vec{x}_1 \neq \varepsilon$, some $(k_1,k_2)\in\mg(M)$ satisfies $k_1>0$. Therefore, the word $\omega_0\,\sigma_1^{k_1}\,\omega_1\, \sigma_2^{k_2}\,\omega_2 \in T\subseteq L'$ has the letter $\sigma_1$ immediately after $\omega_0$. If the leftmost non-empty variable block of $p'$ began at a position $d>|\omega_0|+1$, every word of $L'$ would begin with $(\omega_0\omega_1\omega_2)[1:d-1]$, whose letter immediately after $\omega_0$ is $\omega_1[1] \neq \sigma_1$. Therefore, the leftmost non-empty variable block of $p'$ begins at a position $d\le|\omega_0|+1$. By Lemma~\ref{lem:len-twoblock-positions}, $p'$ has the prefix $\omega_0$, so that $d=|\omega_0|+1$. This proves Claim 1.

\emph{Claim 2: If $\vec{x}_2\ne\varepsilon$, then $p'$ ends with $y\omega_2$ for some $y\in X$.} The proof of this claim is symmetric to the proof of Claim 1.


\smallskip

Let us distinguish two cases.

\emph{Case 1: $p$ has two variable blocks.}  Since $p$ and $p'$ have the same terminal word, with Claim 1 and Claim 2 we obtain that $p'$ has two variable blocks, starts with $\omega_0x$ and ends with $y\omega_2$, where $x,y\in X$. This yields $p'=\omega_0 \vec{x}'_1 \omega_1 \vec{x}'_2 \omega_2$ where $\vec{x}'_1,\vec{x}'_2$ are non-empty variable blocks. Let $M'$ be the monoid generated by the set of gap vectors of $p'$. Since $\varepsilon L_{len}(p')$ contains every word of the form $\omega_0\,\sigma_1^{k_1}\,\omega_1\, \sigma_2^{k_2}\,\omega_2$ where $(k_1,k_2) \in \mg(M)$, we obtain $\mg(M)\subseteq M'$, and thus $M\subseteq M'$. This implies $\varepsilon L_{len}(p)\subseteq\varepsilon L_{len}(p')$, i.e., $\varepsilon L_{len}(p')=\varepsilon L_{len}(p)$.

\emph{Case 2: $p$ has exactly one variable block.} By our normalization, $\vec{x}_1\ne\varepsilon=\vec{x}_2$. 
Then $p=\omega_0\vec{x}_1\omega_1\omega_2$. By Claim 1, $p'$ starts with $\omega_0\vec{x}'_1$. Specifically, $p'=\omega_0\vec{x}'_1\mu_1\vec{x}'_2\mu_2$, where $\mu_1\mu_2=\omega_1\omega_2$. If $p'$ has two variable blocks, i.e., $\vec{x}'_2\ne\varepsilon$ and $\mu_1\ne\varepsilon$, then $p'$ can generate strings not ending in $\omega_1\omega_2$, contradicting $\varepsilon L_{len}(p')\subseteq \varepsilon L_{len}(p)$. So $p'$ has only one variable block, and $p'=\omega_0\vec{x}'_1\omega_1\omega_2$. Let $M'$ be the monoid generated by the set of gap vectors of $p'$. As in Case 1, we obtain $\mg(M)\subseteq M'$, and thus $M\subseteq M'$, which finally yields $\varepsilon L_{len}(p')=\varepsilon L_{len}(p)$.
\end{proof}


We are now ready to prove Theorem~\ref{thm:twoblocks}.

\begin{proof} \emph{(of Theorem~\ref{thm:twoblocks}.)}
Immediate from Lemma~\ref{lem:len-twoblock-telltale} and Angluin's characterization, which states that an indexed family $(L_i)_{i\in\mathbb{N}}$ of languages is learnable in the limit from positive data if and only if there is an effective procedure that, given an index $i$, enumerates a telltale for $L_i$ with respect to $\{L_i\mid i\in\mathbb{N}\}$~\cite{Angluin80telltales}. Here note that the telltales in Lemma~\ref{lem:len-twoblock-telltale} can obviously be enumerated effectively, since the set of minimal generators of the underlying monoid $M$ is a finite subset of the corresponding set of gap vectors.
\end{proof}

It is worth noting that the telltales of Lemma~\ref{lem:len-twoblock-telltale} are of size linear in the length of the underlying patterns, both in terms of the number of words in the telltales, and in terms of the length of these words. 

\section{Iterative Learning of Patterns with at Most One Variable Block}
\label{sec:len-oneblock}

An iterative learner computes each hypothesis from its previous hypothesis and the current example only, ignoring the examples seen in previous steps~\cite{Wiehagen76}. Lange and Wiehagen~\cite{LangeW91} showed that Angluin's pattern languages can be learned iteratively in the limit from positive data. We now show that the equal-length pattern languages with at most one variable block can be learned iteratively as well, without any memory beyond the hypothesis, and that this holds even over a binary alphabet.

\medskip

\begin{definition}[\cite{Wiehagen76}]
    A learner $A$ is iterative if, for any input $(t_0,t_1,\ldots,t_n)$ with $n>0$, the output $A(t_0,t_1,\ldots,t_n)$ depends only on $t_n$ and $A(t_0,t_1,\ldots,t_{n-1})$.
\end{definition}

\medskip

Let $\eLlen^{1}$ be the class of the languages $\varepsilon L_{\mathrm{len}}(p)$ of patterns of the form $p=\omega\,\vec{x}\,\overline{\omega}$ with $\omega,\overline{\omega}\in\Sigma^\ast$, $\vec{x}\in X^\ast$, and $p\ne\varepsilon$, i.e., patterns with at most one variable block. If the variables in $p$ occur with multiplicities $g_1,\dots,g_k$ and $M=\langle g_1,\dots,g_k\rangle\subseteq\N$ is the numerical monoid they generate, Observation~\ref{lem:len-normal} gives
\begin{equation}\label{eq:len-oneblock}
  \varepsilon L_{\mathrm{len}}(p)=\{\omega\,u\,\overline{\omega}\mid u\in\Sigma^\ast,\ |u|\in M\}=:\omega\,\Sigma^{M}\,\overline{\omega}.
\end{equation}
Every finite set of positive integers occurs as the set of variable multiplicities of such a pattern. Note that, within this section, $M$ is always a submonoid of $\mathbb{N}$.

Given an unknown target pattern $p$ of this form, a learner has to identify three items: the prefix $\omega$, the suffix $\overline{\omega}$, and the monoid $M$. The prefix and the suffix are the longest common prefix and suffix of all examples, and these can be updated one example at a time. Once they are correct, the length of every example minus $|\omega\overline{\omega}|$ is an element of $M$, and the learner collects these elements, keeping only the minimal generators. All of this information can be stored in the hypothesis itself, which can then be updated as needed, so that the learner is iterative.

\paragraph{Notation.} For words $u,w$ let $\lcp(u,w)$ and $\lcs(u,w)$ be their longest common prefix and suffix, and for a set $D$ of words let $\lcp(D)$ and $\lcs(D)$ be the longest common prefix and suffix of all words in $D$. Clearly,  \begin{equation}\label{eqn:prefixsuffix}
    (L=\omega\Sigma^M\overline{\omega}\mbox{ with }M\neq\{0\})\mbox{ implies }(\lcp(L)=\omega\mbox{ and } \lcs(L)=\overline{\omega}).
\end{equation}

\paragraph{The learner $A_{\mathrm{it}}$.} In step 0, on the first example $t_0$, the learner $A_{\mathrm{it}}$ outputs the pattern $p_0:=t_0$. 
In step $n\ge 1$, the learner has access to input datum $t_n$ and its previous hypothesis, which is a pattern $p_{n-1}$.
\begin{enumerate}
\item If $p_{n-1}=t_n$, the learner outputs $p_n:=p_{n-1}=t_n$.\\
(Note that the learner will output a hypothesis $w\in\Sigma^*$ if and only if all examples seen so far equal $w$.)
\item If $p_{n-1}\in\Sigma^*$, but $p_{n-1}\ne t_n$, then let $r:=|t_n|-|\lcp(p_{n-1},t_n)|-|\lcs(p_{n-1},t_n)|$. The learner \emph{resets the (currently empty) variable block}:
\begin{itemize}
    \item If $r>0$, then the learner outputs  $p_n:=\lcp(p_{n-1},t_n)x_1^r\lcs(p_{n-1},t_n)$; this entails the conjecture that $M=\langle r\rangle$.
    \item  If $r\le 0$, then the learner outputs  $p_n:=\lcp(p_{n-1},t_n)x_1x_2\lcs(p_{n-1},t_n)$. 
\end{itemize}
(Note that, in case $r=1$, the hypothesis $p_n:=\lcp(p_{n-1},t_n)x_1\lcs(p_{n-1},t_n)$ is equivalent to a hypothesis of the form $\lcp(p_{n-1},t_n)x_1x_2\lcs(p_{n-1},t_n)$ that may be output in some situation where $r\le 0$, but the learner uses two different representations of the same language to be able to distinguish these two cases in step $n+1$. This is the only form in which the learner uses two representations of the same language.)
\item If $p_{n-1}=v\vec{x}\bar{v}$ for $v,\bar{v}\in\Sigma^*$ and $\vec{x}\in X^+$, then let $v':=\lcp(v,t_n)$, $\bar{v}':=\lcs(\bar{v},t_n)$, and $r:=|t_n|-|v'|-|\bar{v}'|$. 
\begin{itemize}
\item If $(v',\bar{v}')=(v,\bar{v})$, and $\vec{x}=x_1x_2$, then the learner outputs (i) $p_n:=v'\,x_1^r\,\bar{v}'$, if $r>0$, which entails the conjecture that $M=\langle r\rangle$, and outputs (ii) $p_n:=v'\,x_1x_2\,\bar{v}'$, if $r\le 0$.
\item If $(v',\bar{v}')=(v,\bar{v})$, $\vec{x}\ne x_1x_2$, and $\vec{x}=x_1^{g_1}\cdots x_k^{g_k}$ where $k>0$ and $g_i>0$ for $i\in\{1,\ldots,k\}$, then (i)~if $r\le 0$, the learner outputs $p_n:=p_{n-1}$; (ii)~if $r>0$, the learner outputs $p_n:=v'\,x_1^{g'_1}x_2^{g'_2}\cdots x_l^{g'_l}\,\bar{v}'$, where $\{g'_1,g'_2\cdots g'_l\}=\mg(\{g_1,\ldots,g_k,r\})$ and $g'_1<g'_2<\ldots <g'_l$. This entails the conjecture that $M=\langle g_1,\ldots,g_k,r\rangle$.
\item If $(v',\bar{v}')\ne (v,\bar{v})$, then the learner \emph{resets the variable block}:
\begin{itemize}
    \item If $r>0$, then the learner outputs  $p_n:=v'x_1^r\bar{v}'$; this entails the conjecture that $M=\langle r\rangle$.
    \item  If $r\le 0$, then the learner outputs  $p_n:=v'x_1x_2\bar{v}'$. 
\end{itemize}
\end{itemize}
\end{enumerate}
By definition, $A_{\mathrm{it}}$ is an iterative learner. It should be noted that hypotheses of the form $v\,x_1x_2\,\bar{v}$ are the only ones carrying information beyond the language of the hypothesis; they mark transient states at which the learner knows that at least one variable must occur in the target pattern, but nothing is known about the multiplicities of variables as of yet. Hypotheses of this form are never used in the limit; i.e., the learner's output on a valid text for a language in $\eLlen^{1}$ will not converge to such a hypothesis (see the proof of Theorem~\ref{thm:len-oneblock} for details).

\medskip

\begin{example}
Let $L=ab\,\Sigma^{\langle2,3\rangle}\,b$ over $\Sigma=\{a,b\}$, and consider a text beginning with $abaab$, $abbbbb$, $abaaab$, $ababb$.
\begin{itemize}
  \item After $abaab$ the learner $A_{\mathrm{it}}$ outputs $abaab$.
  \item After $abbbbb$: $\lcp(abaab,abbbbb)=ab$, $\lcs(abaab,abbbbb)=b$, and $r=6-2-1=3$. The learner outputs $ab\, x_1^3\, b$.
  \item After $abaaab$: the prefix and suffix do not change ($v=v'=ab$, $\bar{v}=\bar{v}'=b$) and $r=3$, so the output stays $ab\, x_1^3\, b$.
  \item After $ababb$: the prefix and suffix do not change ($v=v'=ab$, $\bar{v}=\bar{v}'=b$) and $r=2$. Since $\mg(\{2,3\})=\{2,3\}$, the learner $A_{\mathrm{it}}$ outputs $ab\,x_1^2x_2^3\,b$, which generates $L$.
\end{itemize}
\end{example}

\begin{theorem}\label{thm:len-oneblock}
Let $|\Sigma|\ge2$. Then $A_{\mathrm{it}}$ learns $\eLlen^{1}$ in the limit from positive data. In particular, $\eLlen^{1}$ is learnable in the limit from positive data by an iterative learner.
\end{theorem}

\begin{proof}
Let $p=\omega\vec{x}\bar{\omega}$ be the target pattern, and let $s:=\omega\bar{\omega}$. Let $D_n=\{t_0,t_1,\ldots,t_n\}$ be the set of examples collected up until step $n$.

\emph{Invariant.} By induction on $n$, the following facts are immediate from the definition of $A_{\mathrm{it}}$:
\begin{enumerate}
    \item If $A_{\mathrm{it}}$'s hypothesis at the end of step $n$  is $w\in\Sigma^*$, then $D_n=\{w\}$.
    \item If $A_{\mathrm{it}}$'s hypothesis at the end of step $n$ is $vx_1x_2\bar{v}$, then $v=\lcp(D_n)$, $\bar{v}=\lcs(D_n)$ and, if $n_0$ is the last time step at which $A_{\mathrm{it}}$ reset the variable block, then $|t_j|-|v|-|\bar{v}|\le 0$ for $n_0\le j\le n$. 
    \item If $A_{\mathrm{it}}$'s hypothesis at the end of step $n$  is $vx_1^{g_1}\cdots x_k^{g_k}\bar{v}$, then $v=\lcp(D_n)$, $\bar{v}=\lcs(D_n)$ and, if $n_0$ is the last time step at which $A_{\mathrm{it}}$ reset the variable block, then $\langle g_1,\ldots, g_k \rangle$ is generated by the strictly positive values $|t_j|-|v|-|\bar{v}|$ for $n_0\le j\le n$. 
\end{enumerate} 


Observe that the learner $A_{\mathrm{it}}$ resets the variable block only a finite number of times. This is true because, after the first reset the learner never outputs a word in $\Sigma^*$, and every further reset strictly shortens $v$ or $\bar{v}$, both of which have length at most $|t_0|$.

Let the target be $L=\omega\Sigma^M\overline{\omega}$ and recall $s=\omega\overline{\omega}$. If $M=\{0\}$, then $L=\{s\}$, every example equals $s$, and the learner outputs $s$ forever.

Let $M\neq\{0\}$. Then $L$ is infinite, so the learner must reset the variable block at least once. Since $A_{\mathrm{it}}$ resets the variable block only a finite number of times, there is some time step $n_0$ at which the last reset has occurred. 
From then on, $v$ and $\bar{v}$ do not change. By the invariant, we have $v=\lcp(D_n)=\lcp(D_{n_0})=\omega$ and $\bar{v}=\lcs(D_n)=\lcs(D_{n_0})=\bar{\omega}$ for all $n\ge n_0$. (Note that the text eventually contains two words whose $\lcp$ is exactly $\omega$, which exist by (\ref{eqn:prefixsuffix}) since $M\ne\{0\}$. Likewise, the text eventually contains two words whose $\lcs$ is exactly $\bar{\omega}$.) Consequently, $|t_j|-|v|-|\bar{v}|=|t_j|-|s|\in M$ for every example $t_j$. Hence, from step $n_0$ onward, the learner outputs a pattern of the form $\omega\vec{x}'\bar{\omega}$ such that either $\vec{x}'=x_1x_2$ or $\varepsilon L_{\mathrm{len}}(\omega\vec{x}'\bar{\omega})\subseteq \varepsilon L_{\mathrm{len}}(p)$.

We now show the key claim: \emph{for every $m\in M\setminus\{0\}$, some word of $L$ of length $|s|+m$ appears in the text at a step $j\ge n_0$.} Since every word of $L$ appears in the text, it suffices to find a word of $L$ of length $|s|+m$ that is not among the examples before step $n_0$.
\begin{itemize}
  \item If the restart at $n_0$ was the first one, all earlier examples equal a single word $w$, while $L$ contains $|\Sigma|^m\ge2$ words of length $|s|+m$.
  \item Otherwise, the restart at $n_0$ shortened $v$ or $\bar{v}$. Before step $n_0$, the words $v$ and $\bar{v}$ were the longest common prefix and suffix of some words of $L$, so they extended $\omega$ and $\overline{\omega}$, respectively, and at least one of them was a strict extension---say $v$; the case of $\bar{v}$ is symmetric. Let $\sigma$ be the letter of $v$ following $\omega$. All examples before step $n_0$ begin with $\omega \sigma$, but the word $\omega \bar{\sigma}u'\overline{\omega}\in L$ with $\bar{\sigma}\neq \sigma$, $\bar{\sigma}\in\Sigma$, and $|u'|=m-1$ does not.
\end{itemize}

Let $g_1,\dots,g_k$ be the multiplicities of the variables in the target pattern, where we assume without loss of generality that $\{ g_1,\dots,g_k\}=\mg(\{ g_1,\dots,g_k\})$. Note that $\langle g_1,\dots,g_k\rangle=M$. Let $g^{(n)}_1,\dots,g^{(n)}_{k^{(n)}}$ be the multiplicities of the variables in the pattern output by $A_{\mathrm{it}}$ in step $n$. By the claim and the invariant, there is $n_1\ge n_0$ such that $\langle g^{(n)}_1,\dots,g^{(n)}_{k^{(n)}}\rangle$ contains $g_1,\dots,g_k$ for all $n\ge n_1$, and therefore $\langle g^{(n)}_1,\dots,g^{(n)}_{k^{(n)}}\rangle=M$. Since $\mg$ depends only on the generated monoid, $\{g^{(n)}_1,\dots,g^{(n)}_{k^{(n)}}\}=\mg(M)$ for $n\ge n_1$. So the learner outputs a fixed pattern that generates $L$ from time step $n_1$ onward.
\end{proof}


\section{Conclusions}
\label{sec:len-summary}

This paper resolved an open question on the learnability of erasing equal-length relational pattern languages in the limit. In particular, it was shown that, over a binary alphabet, $\eLlen$ is not learnable from positive data: the language of the relational pattern $p:=x_1\,a\,a\,x_1\,x_2x_2\,b\,x_3$ is the union of a strictly increasing chain of languages of the class and thus has no telltale (Theorem~\ref{thm:len-no-telltale}). Since $p$ as well as all the underlying patterns in this chain have exactly three variable blocks, our non-learnability result already holds for the subclass of languages generated by patterns with exactly three variable blocks. 

This result was contrasted with two positive learnability results for structurally much simpler subclasses:
\begin{itemize}
\item The subclass of languages generated by patterns with at most two variable blocks is learnable in the limit from positive data, for every alphabet with at least two letters (Theorem~\ref{thm:twoblocks}).
 \item The subclass of languages generated by patterns with at most one variable block is learnable in the limit from positive data by an iterative learner, for every alphabet with at least two letters (Theorem~\ref{thm:len-oneblock}).
\end{itemize}

Given that (i) for alphabets of size at least three, the whole class of erasing equal-length relational pattern languages is learnable in the limit from positive data  \cite{MousawiZ25}, and (ii) a wide variety of erasing equal-length relational pattern languages are known for which telltales exist even in case of binary alphabets \cite{MousawiZ25}, a natural next step is to characterize those patterns whose equal-length languages over binary alphabets have telltales.

\section*{Disclosure}
Generative AI tools (specifically Anthropic's Claude Opus 5.5) were used to help with working out details of some of our technical proofs, as well as with the proofreading. The authors assume full responsibility for the entire content and presentation of the paper.

\section*{Acknowledgements}
S.\ Zilles was supported through a Canada CIFAR AI Chair at the Alberta Machine Intelligence Institute (Amii), through an NSERC Canada Research Chair, through the New Frontiers in Research Fund under grant no.\ NFRFE-2023-00109 and through the NSERC Discovery Grants program under application no.\ RGPIN-2026-05856.

 \bibliography{bib}
\end{document}